\documentclass{article}[10pt]
\usepackage{spconf,amsmath,graphicx, amsthm, amssymb}

\usepackage{hyperref}
\usepackage{xcolor}
\usepackage{subcaption}

\newtheorem{theorem}{Theorem}
\newtheorem{corollary}{Corollary}

\newtheorem{lemma}{Lemma}
\newtheorem{prop}{Proposition}

\newcommand{\myparagraph}[1]{\vspace{6pt} \noindent \textbf{#1}}

\title{Stability of entropic Wasserstein barycenters \\and application to random geometric graphs}

\name{M. Theveneau \qquad N. Keriven\thanks{This work was funded by ANR GRandMa, ANR-21-CE23-0006.}}
\address{CNRS, GIPSA-lab, UGA, IRISA, Univ. Rennes 1}

\begin{document}

\maketitle
\begin{abstract}
As interest in graph data has grown in recent years, the computation of various geometric tools has become essential. In some area such as mesh processing, they often rely on the computation of geodesics and shortest paths in discretized manifolds. A recent example of such a tool is the computation of \emph{Wasserstein barycenters} (WB), a very general notion of barycenters derived from the theory of Optimal Transport, and their entropic-regularized variant.
In this paper, we examine how WBs on discretized meshes relate to the geometry of the underlying manifold. We first provide a generic stability result with respect to the input cost matrices. We then apply this result to random geometric graphs on manifolds, whose shortest paths converge to geodesics, hence proving the consistency of WBs computed on discretized shapes.
\end{abstract}

\section{Introduction}

Graphs, and their variants, are becoming increasingly popular in machine learning and signal processing to represent many kinds of data \cite{Hu2020}, from social or computer networks to molecules and proteins, three-dimensional shapes, and so on. In some areas, graphs are usually associated with the representation of an underlying ``geometry'', usually as a latent space \cite{Hoff2002}. For instance, the study of Graph Neural Networks and their variants is at the origin of the very active domain of Geometric Deep Learning \cite{Bronstein2021}, and the analysis of such ``geometric'' (random) graphs and their limit is encountered in many domains of data science \cite{Lovasz2012, Belkin2008, Rosasco2010}.

In the same fashion, Optimal Transport (OT) \cite{Villani2008, peyrebook} is a powerful theory that defines geometrically-meaningful distances and mappings, that can be applied to graph-structured data \cite{kerivenOT, Vayer2019}. A resurgence in data science has recently been experienced, mostly due to novel, efficient computation methods \cite{peyrebook}, for instance based on entropic regularization \cite{Cuturi2013a}. Among the many applications derived from OT, Wasserstein barycenters (WB) \cite{Agueh2011} are powerful tools to compute meaningful geometric means between measures that can represent very general objects. They have found applications in imagery \cite{Solomon2015, Rabin2012a}, statistics \cite{Srivastava2018a}, machine learning \cite{Ho2017}, signal processing \cite{Simou2018, Simou2020} and so on. Moreover, they are also amenable to fast computations \cite{Altschuler2021}, for instance when combined with entropic regularization \cite{Cuturi2014a}.

In this paper, we examine some theoretical properties of Wasserstein barycenters on irregular domains such as (random) graphs, where the ground \emph{cost function} may be noisy and converge to some (unknown) limit. We show that WBs are stable to deformations of the cost matrices that represent the distances in the space, more so when entropic regularization is used. We then apply these results on random geometric graphs, where the shortest paths are known to converge toward the geodesic distances on an underlying manifold. As a result, this guarantees for instance that WBs computed on properly discretized 3D shapes with respect to the shortest paths indeed converge toward the ``true'' WBs (Fig.~\ref{fig:sphere}).

\myparagraph{Outline.} In Sec.~\ref{sec:background}, we start by preliminary materials on OT and Wasserstein barycenters. In Sec.~\ref{sec:wass}, we give a generic stability results of Wasserstein barycenters to deformation cost, before presenting an application on random geometric graphs in Sec.~\ref{sec:rg} with some numerical illustrations. The code to reproduce the figures is available at \url{https://github.com/nkeriven/otrg}. Technical proofs are provided in the Appendix.

\myparagraph{Related Work.} Stability of (classical) OT has been mostly studied w.r.t. the input measures, since an important goal is to understand the convergence speed of OT when replacing the measures by a sampled version \cite{Genevay2018, Mena2019}. There are a few results on the stability w.r.t. cost deformation \cite{Eckstein2021, kerivenOT}, with some applications on random graphs \cite{kerivenOT}. For WBs, stability w.r.t. the input measures has been recently studied \cite{Carlier2022}, but to our knowledge stability w.r.t. cost deformation is novel.

The relationship between shortest paths on geometric graphs and geodesics on manifolds has been long established \cite{Bernstein2000, Arias-Castro2021a}, with many applications in shape and graph analysis \cite{Peyre2010}. OT on shapes has been explored empirically and theoretically \cite{Rabin2010, Solomon2014b, kerivenOT}, and WBs have found applications in imagery, for instance for texture mixing \cite{Rabin2012a}. The theoretical properties of WBs on Riemannian manifolds has been thoroughly explored e.g. in \cite{Kim2017b}, but results pertaining to the infinite-node limit of discretized manifolds such as the one presented here are, to our knowledge, relatively novel.

\myparagraph{Notations.} We define the scalar product between two matrices by $<A, B>=tr(A^t B)$. The probability simplex is $\Delta_+^n = \{a \in \mathbb{R}_+^n; \sum_i a_i = 1\}$. The norm $\|\cdot\|_\infty$ refers to the maximal element both for vectors and matrices. Real functions are applied to vectors and matrices element-by-element, for instance $e^A$ or $\log(A)$.

\begin{figure}[h!]
    \centering
    \begin{subfigure}{.48\textwidth}
    \centering
    \def\sz{.2}
    \includegraphics[width=\sz\textwidth]{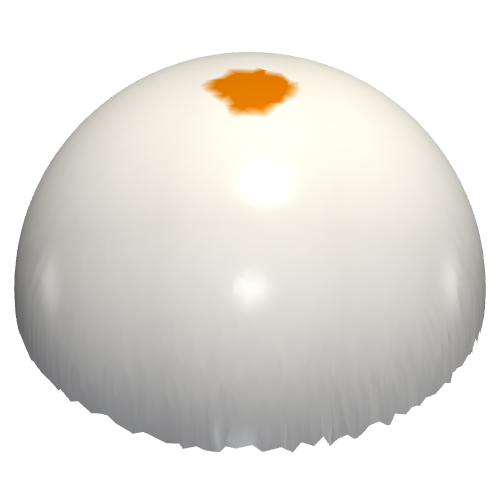}
    \includegraphics[width=\sz\textwidth]{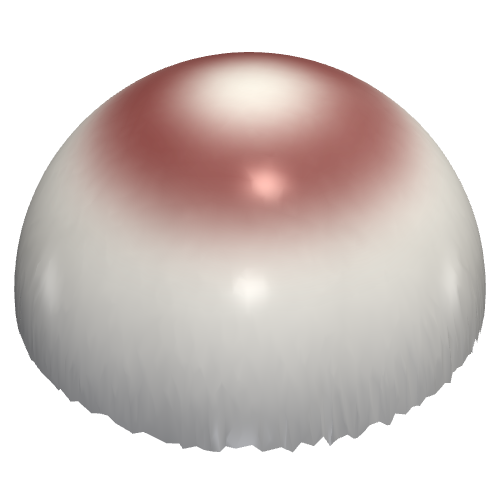}
    \includegraphics[width=\sz\textwidth]{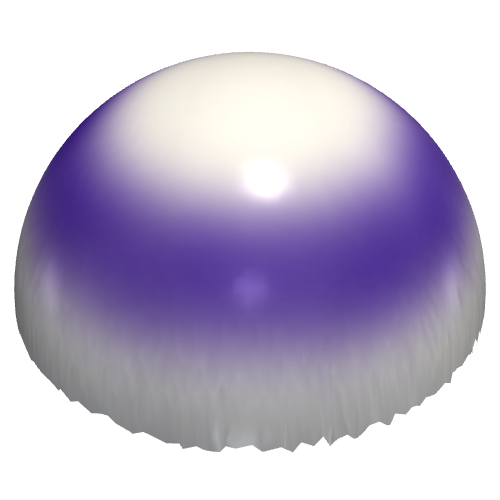}
    \includegraphics[width=\sz\textwidth]{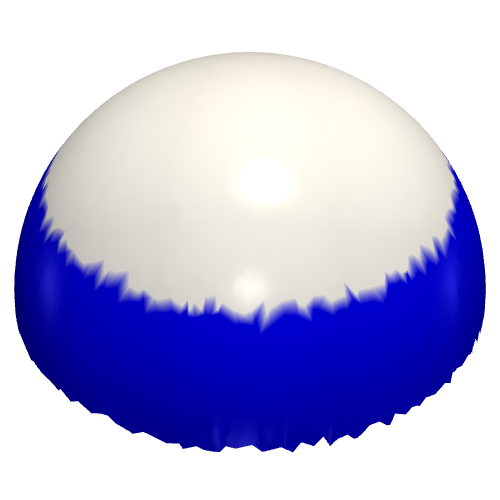} \\
    \includegraphics[width=\sz\textwidth]{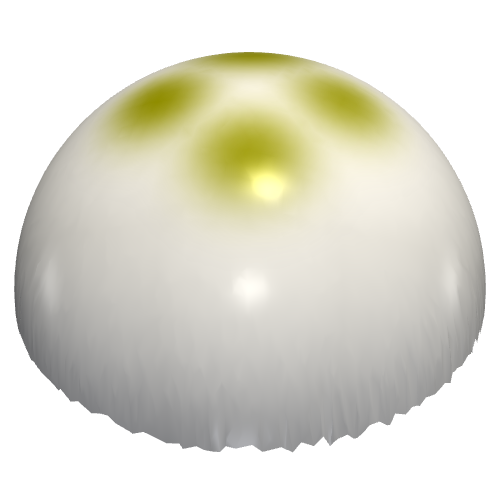}
    \includegraphics[width=\sz\textwidth]{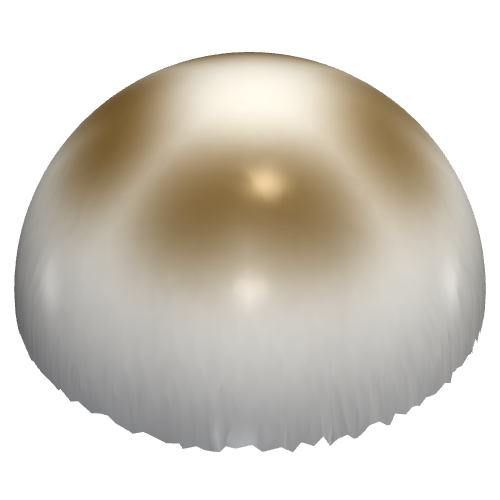}
    \includegraphics[width=\sz\textwidth]{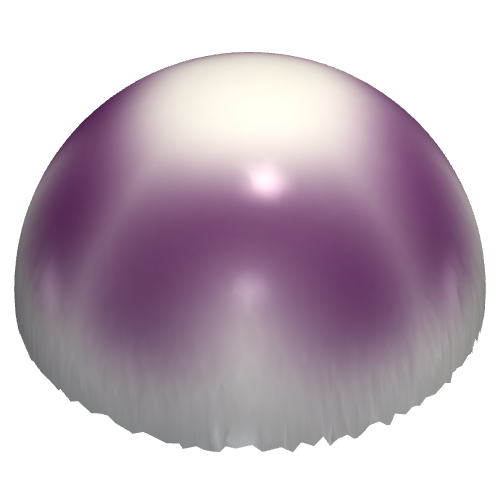}
    \includegraphics[width=\sz\textwidth]{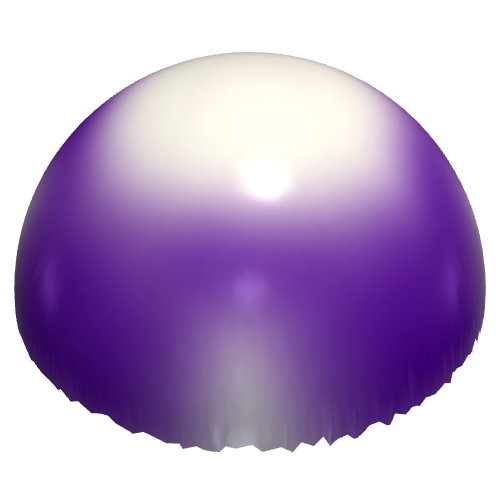} \\
    \includegraphics[width=\sz\textwidth]{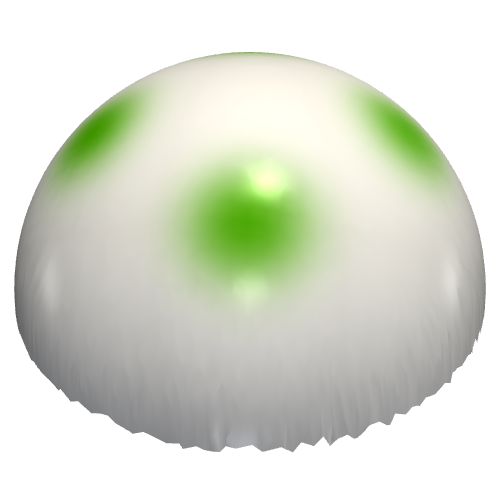}
    \includegraphics[width=\sz\textwidth]{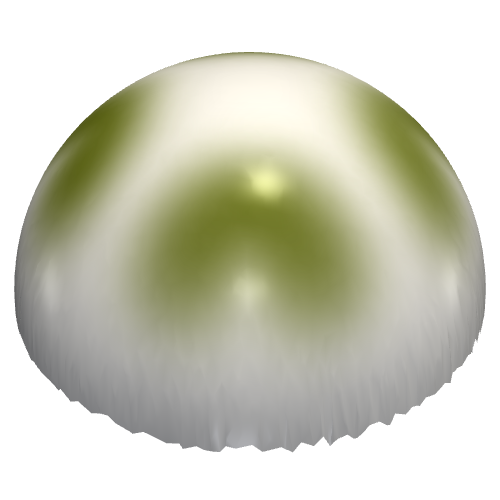}
    \includegraphics[width=\sz\textwidth]{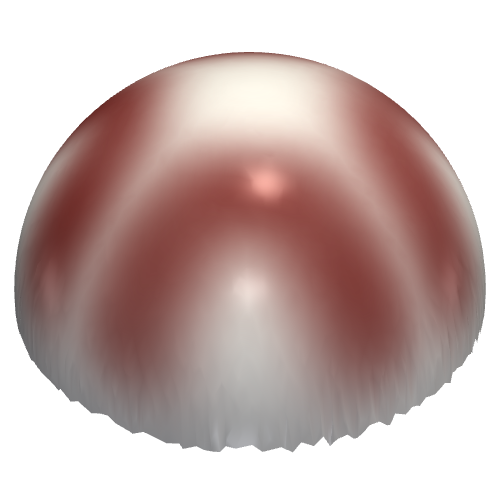}
    \includegraphics[width=\sz\textwidth]{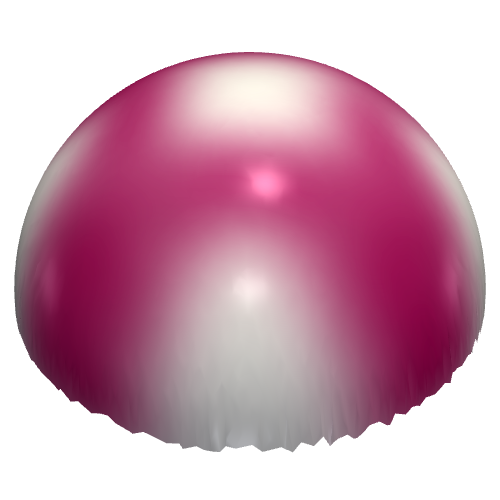} \\
    \includegraphics[width=\sz\textwidth]{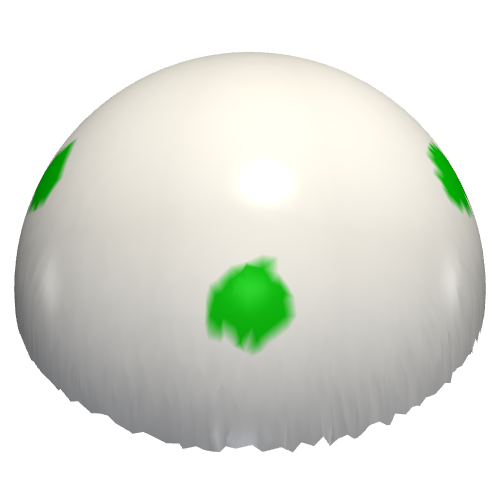}
    \includegraphics[width=\sz\textwidth]{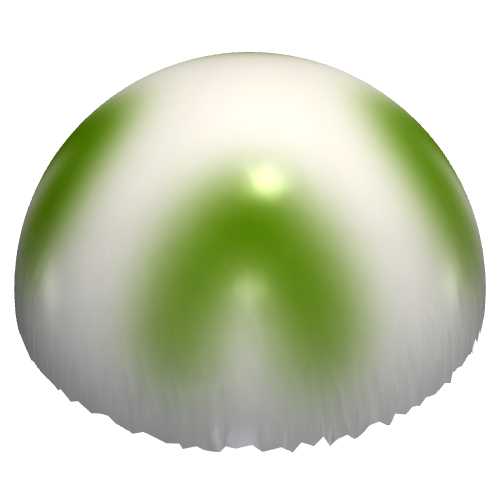}
    \includegraphics[width=\sz\textwidth]{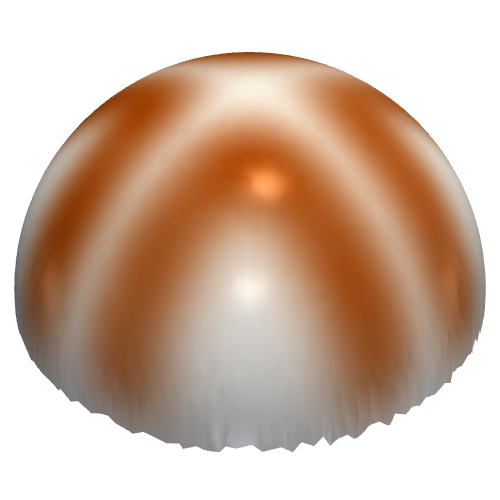}
    \includegraphics[width=\sz\textwidth]{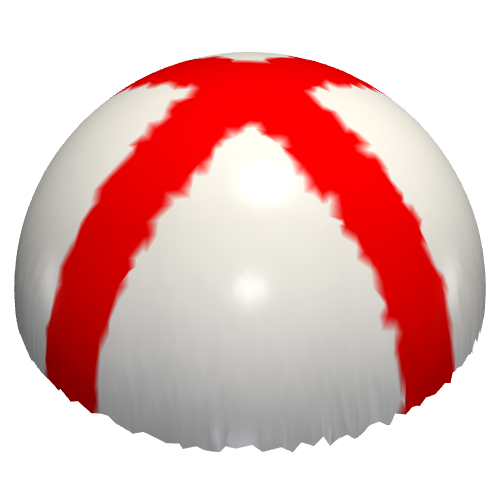}
    \caption{Barycenters with the true geodesics (known for the sphere).}
    \label{subfig:geo}
    \end{subfigure}
    \begin{subfigure}{.48\textwidth}
    \centering
    \def\sz{.2}
    \includegraphics[width=\sz\textwidth]{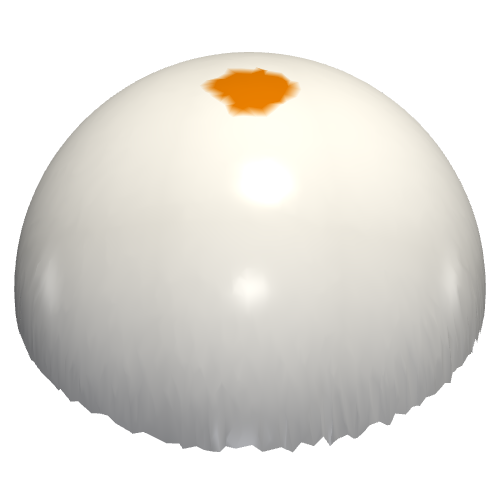}
    \includegraphics[width=\sz\textwidth]{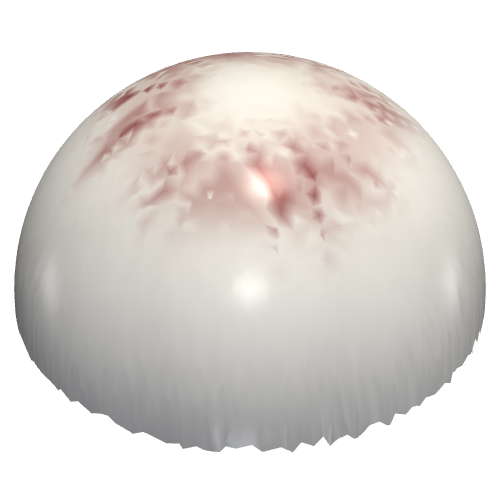}
    \includegraphics[width=\sz\textwidth]{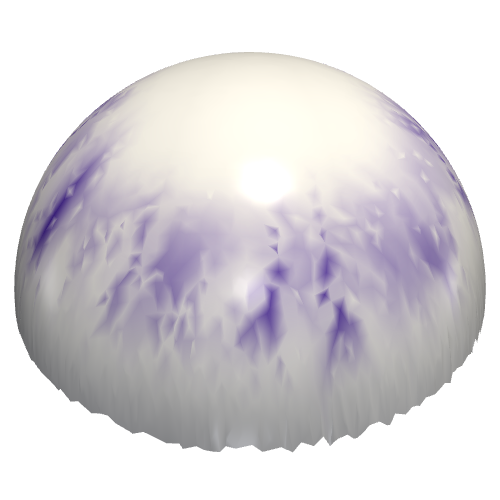}
    \includegraphics[width=\sz\textwidth]{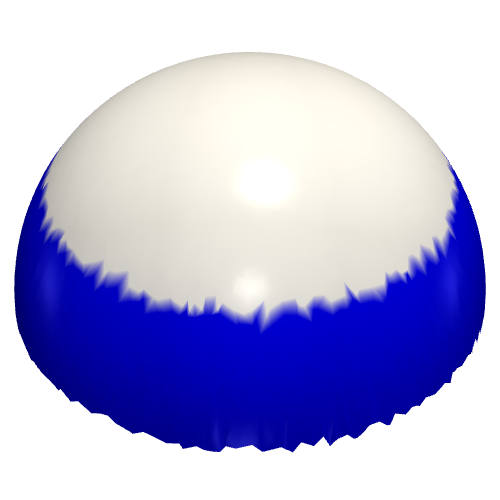} \\
    \includegraphics[width=\sz\textwidth]{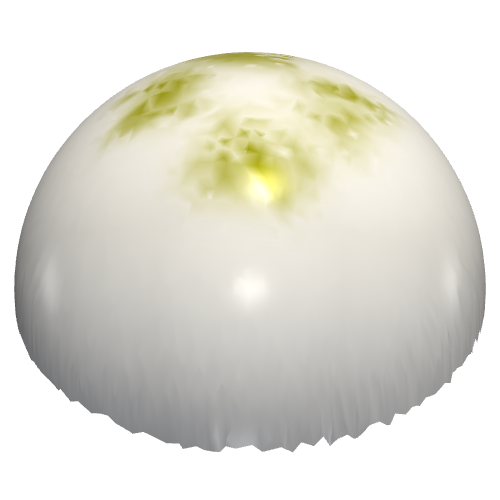}
    \includegraphics[width=\sz\textwidth]{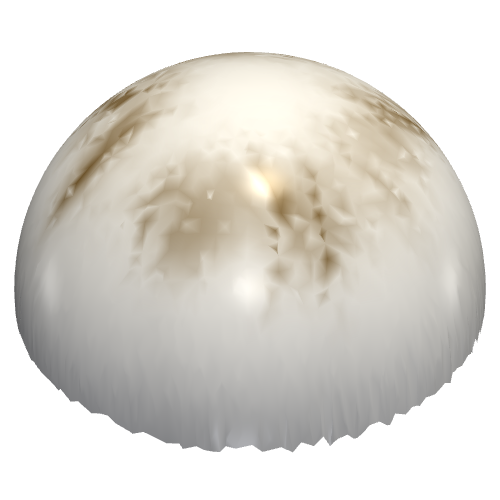}
    \includegraphics[width=\sz\textwidth]{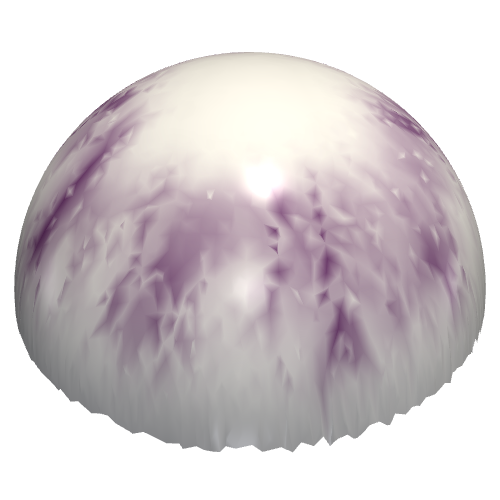}
    \includegraphics[width=\sz\textwidth]{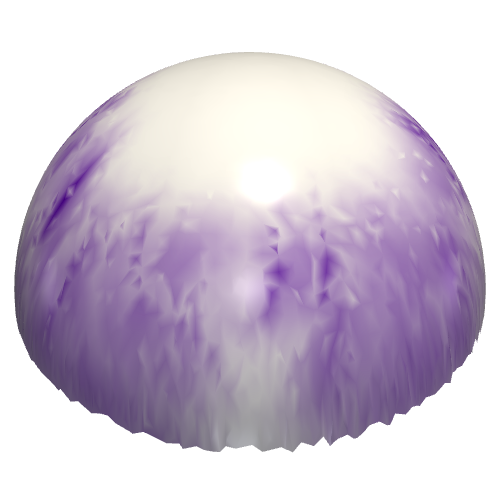} \\
    \includegraphics[width=\sz\textwidth]{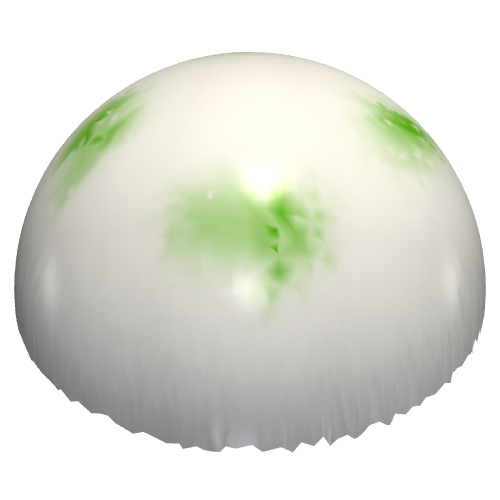}
    \includegraphics[width=\sz\textwidth]{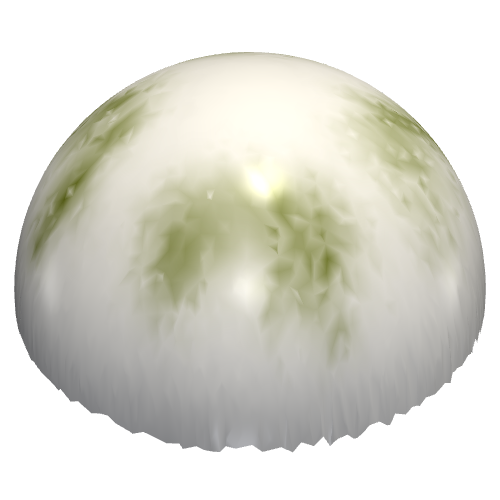}
    \includegraphics[width=\sz\textwidth]{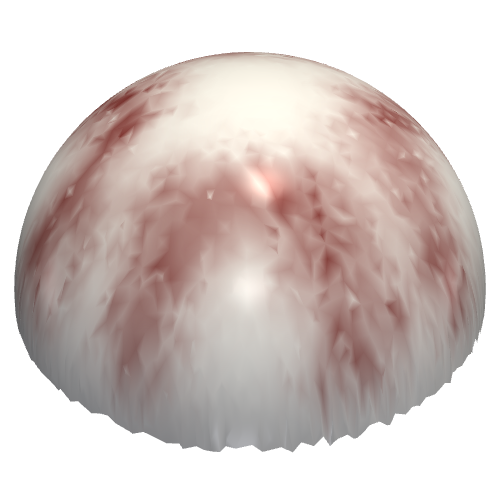}
    \includegraphics[width=\sz\textwidth]{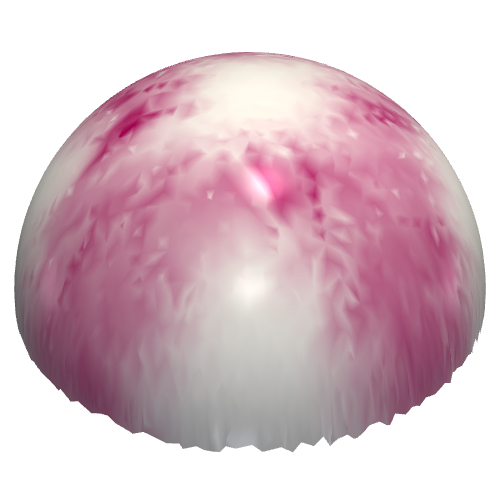} \\
    \includegraphics[width=\sz\textwidth]{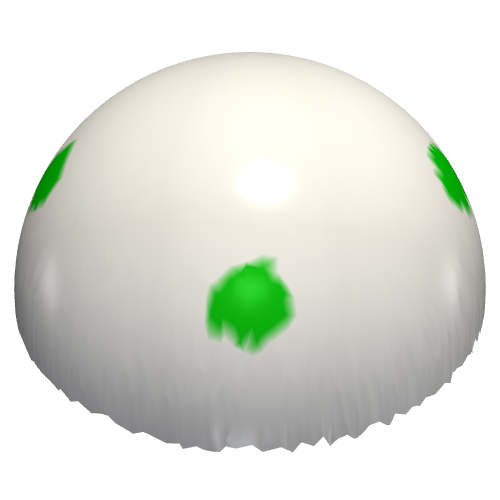}
    \includegraphics[width=\sz\textwidth]{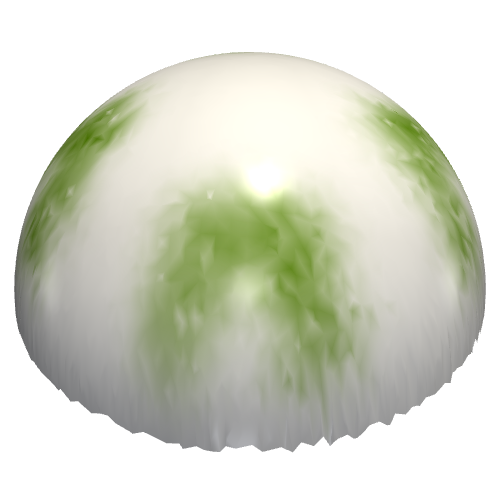}
    \includegraphics[width=\sz\textwidth]{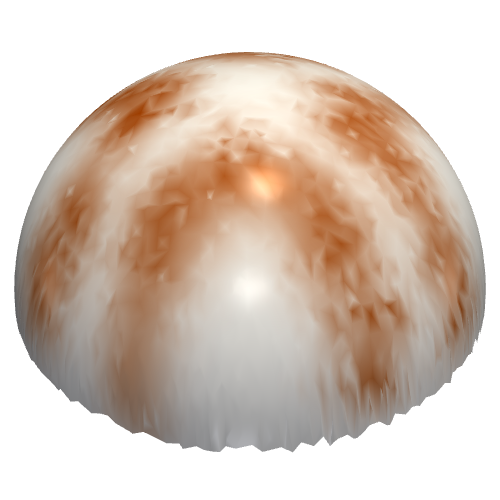}
    \includegraphics[width=\sz\textwidth]{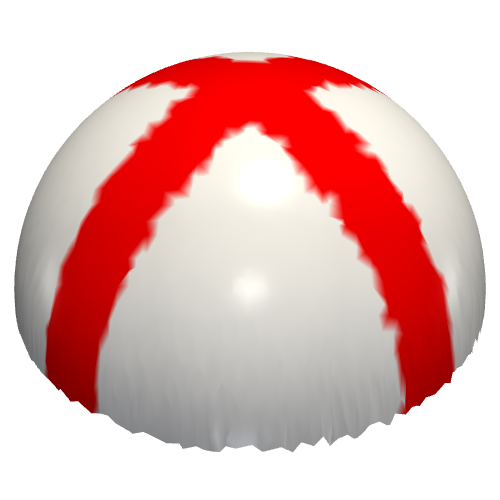}
    \caption{Barycenters with the shortest paths in a random graph.}
    \label{subfig:rg}
    \end{subfigure}
    \caption{Interpolation between $S=4$ distributions on a half sphere, for varying weights $\lambda_s$, with true geodesics (top) or estimated ones on a random graph (bottom). Sec.~\ref{sec:wass} and \ref{sec:rg} prove the convergence of the noisy barycenters to the true ones when the number of nodes goes to infinity.}
    \label{fig:sphere}
\end{figure}

\section{Background on Optimal Transport}\label{sec:background}

Let us start by recalling some background material on discrete Optimal Transport. We consider two finite distributions $a \in \Delta_+^n$ and $b \in \Delta_+^m$ as well as a \emph{cost matrix} $C \in \mathbb{R}_+^{n \times m}$. Usually (but technically not necessarily!), $a$ and $b$ are \emph{weights} associated to two sets of points $\{x_1,\ldots, x_n\}$ and $\{y_1,\ldots,y_m\}$ and $C_{ij}$ indicates how ``costly'' it is to transport mass from $x_i$ to $y_j$, often through a metric $d$ elevated to some power $C_{ij} = d(x_i,y_j)^p$. The $x_i$ and $y_j$ can live in different spaces, as long as $d$ is properly defined.

We denote by $U_{a,b}=\left\lbrace T\in\mathbb{R}_+^{n\times m}, T 1_m=a, T^\top 1_n=b\right\rbrace$ the set of \emph{couplings} between $a$ and $b$. The OT distance between $a$ and $b$ is defined as:
\begin{equation}\label{eq:wass}
W_C(a, b)
:=\min\limits_{T\in U(a,b)}<T, C>\, .
\end{equation}
When $C_{ij} = d(x_i,y_j)^p$, $W_C(a,b)^{1/p}$ is the so-called $p$-Wasserstein distance between the measures $\sum_i a_i \delta_{x_i}$ and $\sum_j b_j \delta_{y_j}$. However, note that \emph{only the knowledge of $a,b,C$} is necessary to compute $W_C(a,b)$.

Computing \eqref{eq:wass} is a linear problem, which makes it difficult to solve at large scale. This can be handled by adding \emph{entropic regularization} to the cost function \cite{Cuturi2013a}: for $\epsilon\geq 0$,
\begin{align}
&W_C^\epsilon(a,b) := \min_{T\in U(a,b)}W_C^\epsilon(a,b, T) \notag \\
&\qquad \textup{where}\quad  W_C^\epsilon(a,b,T) := <T, C> - \epsilon H(T) \label{eq:wass_eps}
\end{align}
where $H(T) = -\sum_{i,j}T_{i,j}\log T_{i,j}$ with the convention that $0 \log 0=0$ by continuity. The resulting problem is strictly convex when $\epsilon >0$, and can be solved efficiently by a numbers of methods \cite{peyrebook}, including the celebrated Sinkhorn's algorithm \cite{Cuturi2013a}. When $\epsilon \to 0$, the problem converges (in various ways) to the unregularized one \eqref{eq:wass} \cite{peyrebook}.

In this paper, we examine so-called \emph{Wasserstein barycenters}. Consider $S$ discrete measures $b_s \in \Delta_+^{m_s}$ of size $m_s$, along with $S$ cost matrices $C_s \in \mathbb{R}_+^{n \times m_s}$ that indicate the transportation cost from each $b_s$ \emph{to a common space of size $n$}. The ``barycenter'' of the $b_s$ is thus a measure $a\in \Delta_+^n$. Given weights $\lambda \in \Delta_+^S$, it is computed by a Fréchet mean w.r.t. the $W^\epsilon$ distance: denoting $\Theta = \{\lambda_s, b_s, C_s\}_{s=1}^S$ for short,
\begin{align}\label{eq:wass_bary}
&B^\epsilon(\Theta) := \min_{a \in \Delta_+^n} B^\epsilon(\Theta, a) \notag \\
&\qquad  \text{where}\quad  B^\epsilon(\Theta, a) := \sum_{s=1}^S \lambda_s W^\epsilon_{C_s}(a, b_s)
\end{align}
This is a smooth convex optimization problem \cite{peyrebook}, with a unique minimizer when $\epsilon>0$, that we denote by $a^\Theta$. When $\epsilon >0$, it can be computed by a variant of Sinkhorn's algorithm \cite[Chap. 9]{peyrebook}.
As before, generally (but, again, not necessarily) $b_s$ represent the weights of a discrete measure over positions $\{y_{1s},\ldots,y_{m_s s}\}$, the sought-after barycenter $a$ is over some positions $\{x_1,\ldots, x_n\}$, and the cost matrices are defined with metrics $C_{i,j,s}=d_s(x_i, y_{j s})^p$. Again, the spaces in which $y_{j s}, x_i$ live need \emph{not} be the same, as long as the metrics $d_s$ are over the appropriate domains.

In the next section, we examine the stability of this problem to perturbations of the cost matrices $C_s \in \mathbb{R}^{m_s \times n}$, before presenting an application on random geometric graphs on manifolds in Sec.~\ref{sec:rg}.

\section{Stability of Wasserstein barycenters}\label{sec:wass}
    
We study the stability of Wasserstein barycenters \eqref{eq:wass_bary} to perturbations of the cost matrices $C_s$. In the rest of the section, we denote $\Theta = \{\lambda_s, b_s, C_s\}_{s=1}^S$ and $\Tilde{\Theta} = \{\lambda_s, b_s, \Tilde{C}_s\}_{s=1}^S$ with the same $\lambda_s, b_s$ but perturbed cost matrices $\Tilde{C}_s$.

Our first result guarantees closeness of the cost function for any regularization level, including $\epsilon=0$. It does \emph{not}, however, guarantee proximity of the optimal barycenters. The proof, presented in the Appendix, 
 is straightforward.

\begin{prop}\label{prop:stability}
For all $\epsilon\geq 0$, we have
\begin{align}
    &\left\vert B^\epsilon(\Theta) - B^\epsilon(\Tilde{\Theta}) \right\vert \leq \sum_s \lambda_s \|C_s - \Tilde{C}_s\|_\infty
\end{align}
\end{prop}
Hence, if all matrices $\Tilde{C}_s$ converge to $C_s$ in $\infty$-norm, the cost functions $B^\epsilon$ converge to one another. However, this proposition does not provide stability of the barycenter $a^\Theta$ itself, which is what interests us in practice. For this we need strict convexity of the problem, which holds only when $\epsilon >0$. The following theorem is then our main result. Recall that $a^\Theta$ is the optimal barycenter in \eqref{eq:wass_bary}.
\begin{theorem}\label{thm:stability}
Assume $0 \leq c_{\min} \leq C_{sij}, \Tilde{C}_{sij} \leq c_{\max}$ hold for all $s,i,j$. For all $\epsilon>0$ we have
\begin{equation}
    \|a^\Theta-a^{\Tilde{\Theta}}\|_2^2  \lesssim \epsilon e^{3(c_{\max} - c_{\min})/\epsilon} \sum_s \lambda_s \|C_s - \Tilde{C}_s\|_\infty
\end{equation}
\end{theorem}
We therefore obtain stability of the optimal barycenters, with a potentially large multiplicative constant in $\epsilon$. We also note that the bound is insensitive to shifting the costs $C_s$ and $\tilde C_s$ by a constant $c$ (which shifts $c_{\min}$ and $c_{\max}$), which is to be expected since this shifts $W_C$ by the same constant and does not affect the transport plans \cite{kerivenOT}.

In the next section, we apply this result to the approximation of Wasserstein barycenters on manifolds. The rest of this section is dedicated to a sketch of proof of this theorem, whose details can be found in Appendix.

\myparagraph{Sketch of proof.}
As usual in convex optimization, we work with the dual problem of \eqref{eq:wass_bary}, which reads \cite[Chap. 9]{peyrebook}
\begin{align*}\label{eq:dual_barycentre}
    &\max\limits_{f_s,g_s} \Bigg\{\mathcal{L}_\Theta(f,g), \sum\nolimits_s \lambda_s f_s = 0 \Bigg\} \\
    &\quad \text{where}\quad \mathcal{L}_\Theta(f,g) := \sum_s \lambda_s \left(b_s^\top g_s-\epsilon (e^{f_s /\epsilon})^\top K_s e^{g_s /\epsilon}\right)
\end{align*}
with $K_s = e^{-C_s/\epsilon}$ and $f_s \in \mathbb{R}^n, g_s \in \mathbb{R}^{m_s}$, and we use the shortcut $f = \{f_s\}_s$ and similarly for $g$. This is a concave maximization problem with linear constraints and a non-empty solution set, hence strong duality holds. Moreover, it is known \cite[Chap. 9]{peyrebook} that the optimal $f^\Theta_s$ and $g^\Theta_s$ are related to the optimal $a^\Theta$ by
\begin{equation}\label{eq:dual_bary}
\forall s,\quad    a^\Theta = \left(\text{diag}(e^{f^\Theta_s/\epsilon})K_s \text{diag}(e^{g^\Theta_s/\epsilon})\right) 1_{m_s}
\end{equation}
and similarly $b_s = \left(\text{diag}(e^{f_s/\epsilon})K_s \text{diag}(e^{g_s/\epsilon})\right)^\top 1_n$.
Our proof is then similar in principle to that of \cite{kerivenOT}.
We start by bounding the dual potentials $f^\Theta_s, g^\Theta_s$. The following Lemma is a simple consequence of first-order conditions.
\begin{lemma}\label{lem:borne}
For all $s$: assuming $0<\delta_{\min} \leq K_{sij}\leq \delta_{\max}$, we have
\begin{equation}
    \frac{1}{\delta_{\max} }\leq \sum\limits_{i,j} e^{(f^\Theta_{si}+g^\Theta_{sj}) /\epsilon}  \leq  \frac{1}{\delta_{\min} }
\end{equation}
\end{lemma}

We can then use the strict convexity of $\mathcal{L}$ to obtain the following bound.
\begin{lemma}\label{lem:stability}
It holds that
\begin{align*}
    &\sum_{s,i,j}\lambda_s K_{sij}\left|f_{si}^{\Tilde{\Theta}}+g_{sj}^{\Tilde{\Theta}}-(f_{si}^\Theta+g_{sj}^\Theta)\right|^2\\
    &\qquad\qquad \leq 2\epsilon \delta_{\max}\left(\mathcal{L}_\Theta(f^\Theta,g^\Theta)-\mathcal{L}_\Theta(f^{\Tilde{\Theta}},g^{\Tilde{\Theta}})\right)
\end{align*}
\end{lemma}

Finally, using Lemma \ref{lem:borne} we upper bound the right hand side above to get an expression that depends on $\|e^{-C_s/\epsilon} - e^{-\Tilde{C}_s/\epsilon}\|_\infty$, which we bound by $\epsilon e^{-c_{\min}/\epsilon} \|C_s - \Tilde{C}_s\|_\infty$ using the mean value theorem. For the left hand side, we obtain a difference between $a_\Theta$ and $a_{\Tilde{\Theta}}$ using \eqref{eq:dual_bary}, which concludes the proof. The details can be found in the Appendix.

%

\section{Application: Random geometric graphs}\label{sec:rg}

\begin{figure*}[h]
    \centering
    \def\sz{.19}
    \includegraphics[width=\sz\textwidth]{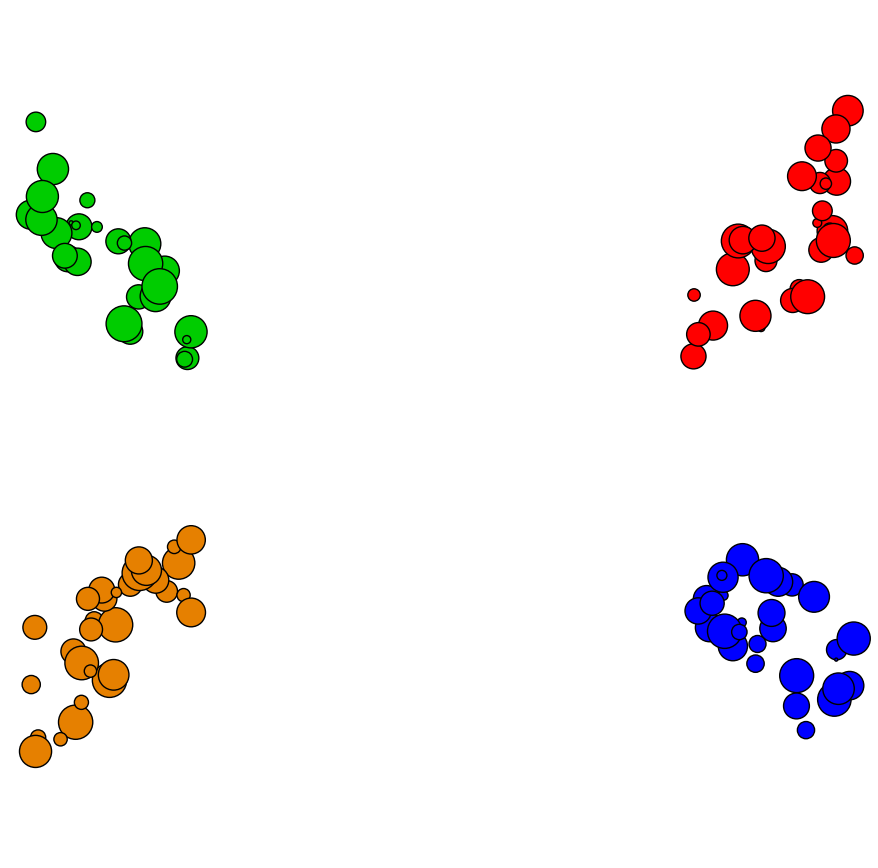}
    \includegraphics[width=\sz\textwidth]{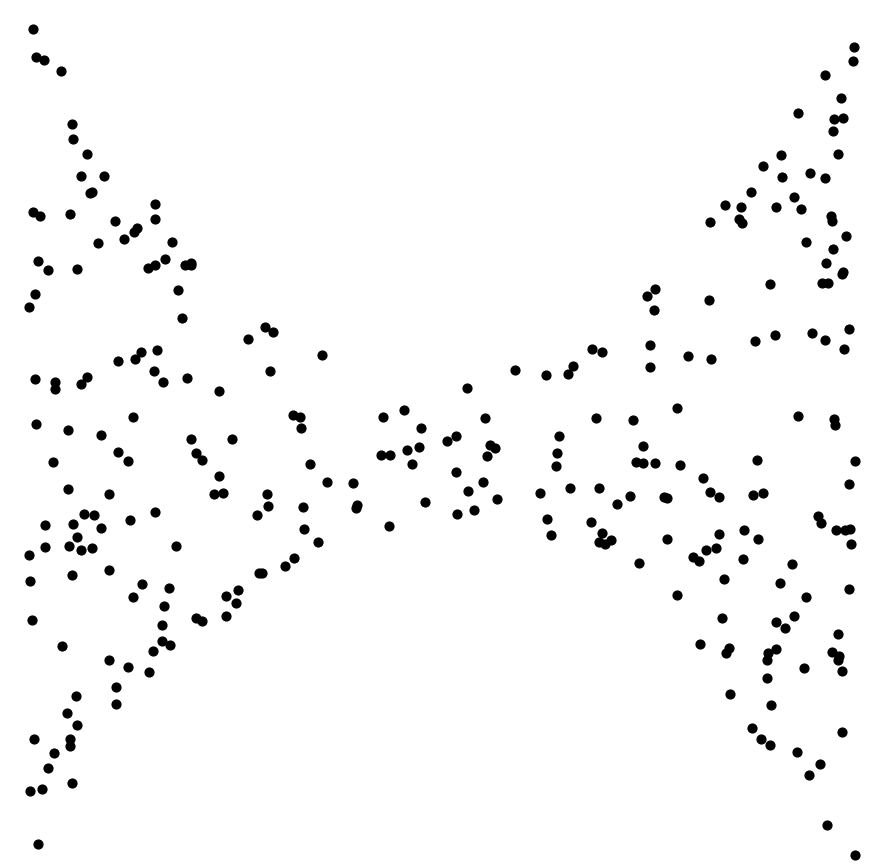}
    \includegraphics[width=\sz\textwidth]{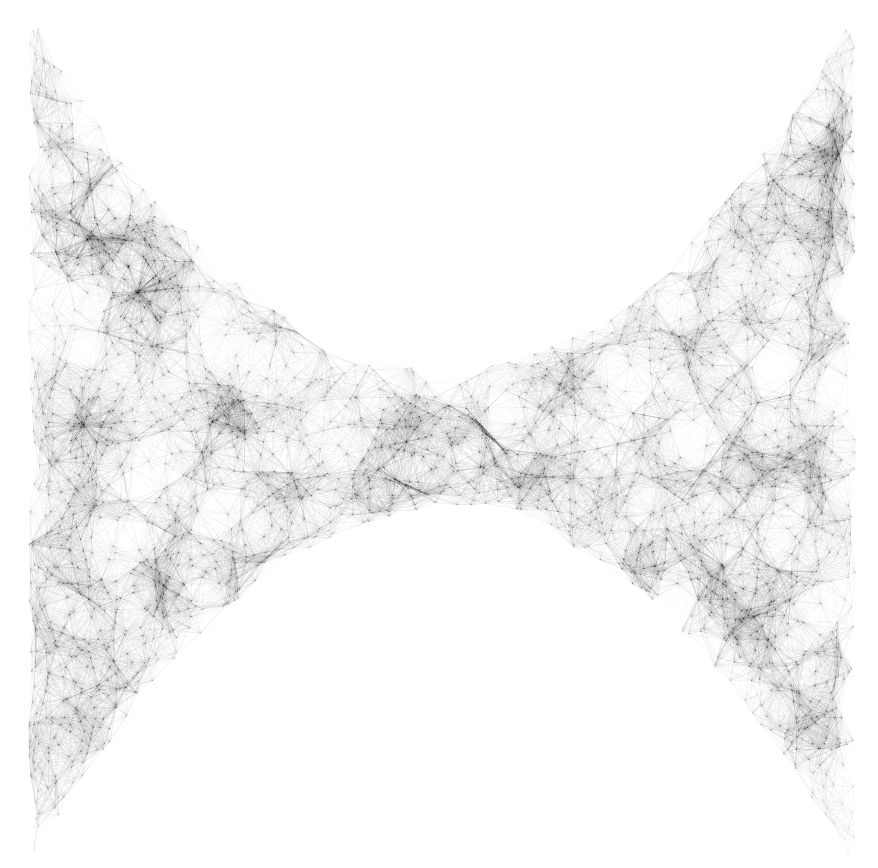}
    \includegraphics[width=\sz\textwidth]{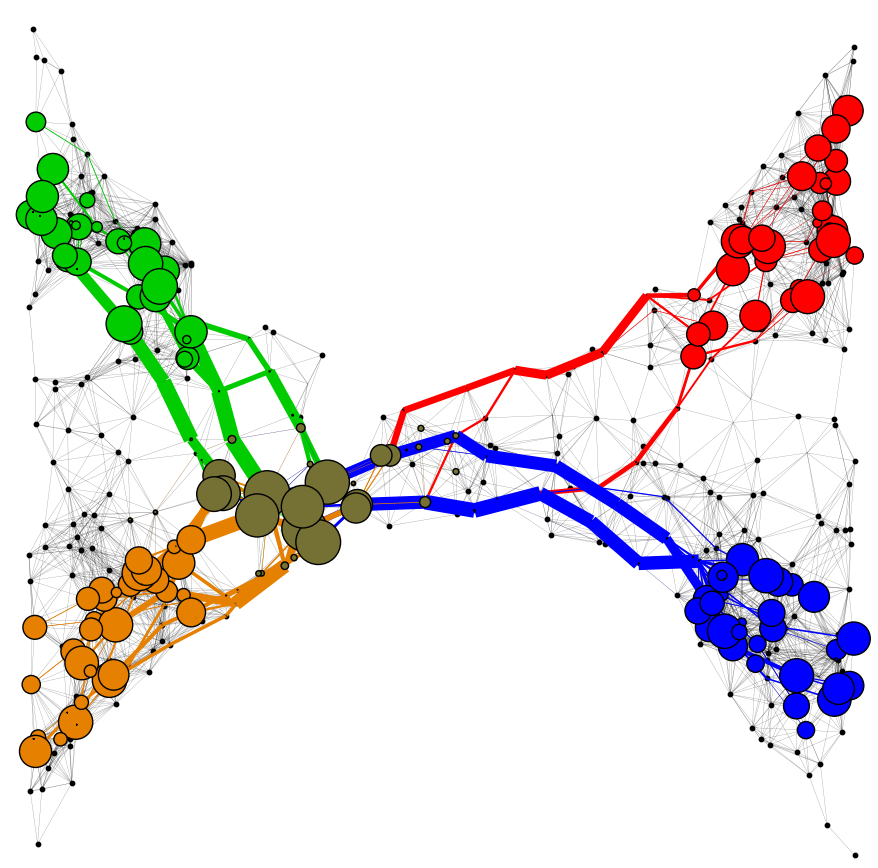}
    \includegraphics[width=\sz\textwidth]{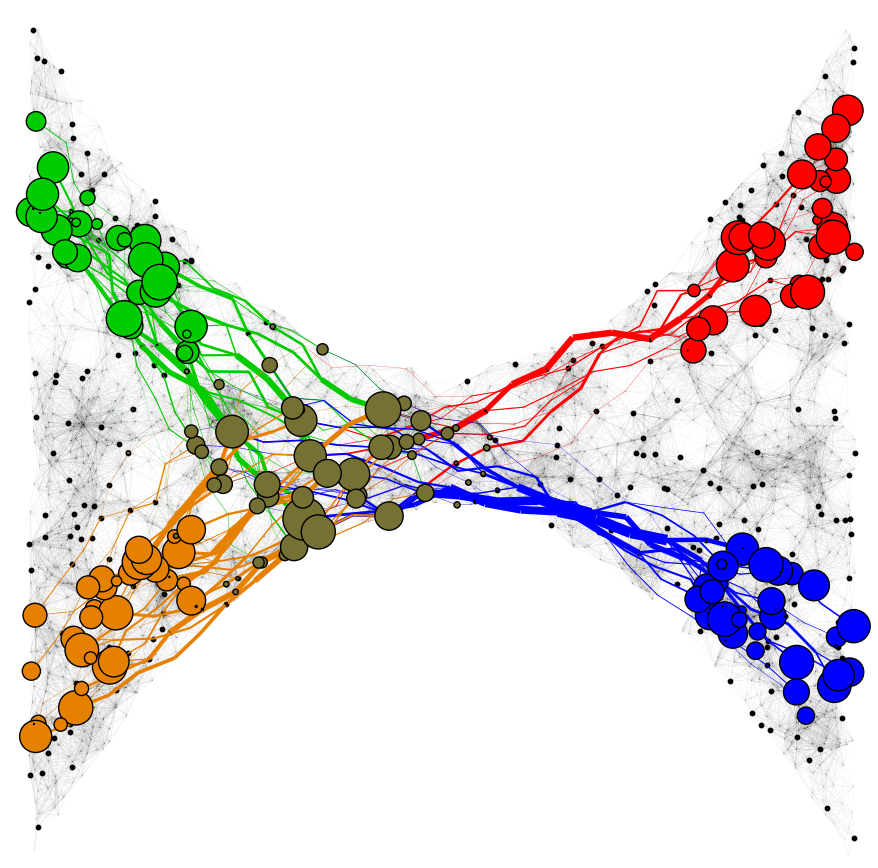}
    \caption{Example of Wasserstein barycenters in a random graph on a 2D domain. From left to right: distributions $\nu_s$ located at points $y_{sj}$; support points for the barycenter $x_i$; random geometric graph constructed with $x_{i}$, $y_{sj}$, as well as additional random points $z_1,\ldots, z_N$; visualization of the barycenters and the transport plans for two values $N=10$ and $N=2000$.}
    \label{fig:tube}
\end{figure*}


A classical approach to manipulating manifolds such as 3D shapes is to discretize them, for instance by constructing a random geometric graph \cite{Penrose2008}. 
This is done by randomly drawing $N$ points on the manifold and connecting them if their distance in the ambient Euclidean space is less than a certain $h_N$ which tends to 0 when $N$ tends to infinity. 
It is then known \cite{Arias-Castro2021a, kerivenOT} that the length of the shortest paths in the graph converge, under some conditions, to the geodesic distance of the manifold.

More precisely, assume that we have a compact, smooth submanifold $\mathcal{M} \subset \mathbb{R}^d$ of dimension $k$, without boundary for simplicity. Its geodesic distance is $d(x,y)$, while $\|x-y\|$ refers to the norm in the ambient space $\mathbb{R}^d$. Its diameter is $D_{\mathcal{M}} := \sup_{u,t \in \mathcal{M}} d(u,t)$.

Consider the following objects: for $1\leq s\leq S$, distributions $\nu_s \in \Delta_+^{d_s}$, weights $\lambda_s$, and $m_s$ supporting points $\{y_{s1}, \ldots, y_{s m_s}\} \subset \mathcal{M}$ on the manifold, for each distribution (they need not be distinct). Then, the $n$ supporting points $\{x_1, \ldots, x_n\}\subset \mathcal{M}$ on which we are going to compute the barycenter $a\in \Delta_+^n$. Finally, we complete with $N$ additional points $\{z_1, \ldots, z_N\}\subset \mathcal{M}$ \emph{drawn i.i.d.} according to some probability distribution $P$ on $\mathcal{M}$, which we assume to have a density $p_z$ with respect to the uniform measure on $\mathcal{M}$, bounded away from zero: $p_z(z) \geq c_z >0$. We then construct a random graph with radius $h_N$ on $\mathcal{M}$ using all the points $\mathcal{V} := \{x_i, y_{sj}, z_\ell\}_{ij\ell}$: if any two such points $u,t \in\mathcal{V}$ satisfy $\|u-t\| \leq h_N$, then we add an edge between them. Note that here $x_{i}$ and $y_{sj}$ are deterministic, while $z_\ell$ are random. We let $N \to \infty$, and $h_N \to 0$, and aim to prove that the shortest paths length between $x_i, y_{sj}$ converge to the geodesic distance. See Fig.~\ref{fig:tube}.

For $p\geq 1$, we denote by $C_s = [d(x_{i}, y_{sj})^p]_{ij} \in \mathbb{R}^{n \times m_s}$ the matrices containing the true geodesic distances between our points of interest elevated to some power $p$, with respect to which we want to compute Wasserstein barycenters. We then denote by $SP(u,t)$ the shortest path (minimal number of edges) in the graph between two vertices $u,t$, with $SP(u,t) = +\infty$ if they are not connected. We define
\begin{equation}
    \tilde C_s = [(h_N SP(x_{i}, y_{sj}))^p]
\end{equation}
the matrices containing the shortest paths between $x_{is}$ and $y_j$, normalized by $h_N$. Then, the following result is from \cite{kerivenOT}.
\begin{lemma}[Theorem 2 in \cite{kerivenOT}]\label{lem:geodesic}
Consider $u,t$ two vertices among the fixed points $\{x_{i}, y_{sj}\}$, and $\rho>0$. For $N$ large enough, with probability $1-\rho$, we have
\begin{equation}
    |d(u,t) - h_N SP(u,t)| \lesssim h_N + \left( \frac{\log \frac{1}{h_N \rho}}{c_z N h_N^k}\right)^{1/k}
\end{equation}
where the multiplicative constant depends on the properties of the manifold $\mathcal{M}$.
\end{lemma}

This convergence translates into the convergence of the cost matrix of the graph to the cost matrix of the manifold. Hence, using a union bound and the results of Sec.~\ref{sec:wass}, we immediately obtain the following corollary.

\begin{corollary}\label{cor:stabilityRG}
With probability $1-\rho$, for all $\epsilon>0$ we have
\begin{equation*}
    \|a^\Theta-a^{\Tilde{\Theta}}\|_2^2  \lesssim p D_\mathcal{M}^{p-1} \epsilon e^{\frac{6D_\mathcal{M}}{\epsilon}} \left( h_N + \left( \frac{\log \frac{n \sum_s m_s}{h_N \rho}}{c_z N h_N^k}\right)^{\frac{1}{k}}\right)
\end{equation*}
\end{corollary}

In other words, as long as when $N\to \infty$,
$
h_N \to 0$ 
and $\frac{N h_N^k}{\log (1/h_N)} \to +\infty
$, 
then the barycenters computed using the shortest paths in the graph converge to the barycenters that use the true geodesic distance $d$, see Fig.~\ref{fig:tube}. Note that on a $k$-manifold the average degree of a random geometric graph is $O(N h_N^k)$, so here the average degree needs to increase to $+\infty$ (the graph is not \emph{sparse}), at least by a logarithmic factor.

\myparagraph{Numerical illustration.} In Fig.~\ref{fig:sphere} and \ref{fig:tube}, we illustrate our results on two examples of discretized manifolds: a sphere, where the true geodesics are known and we observe the effect of the discretization, and a 2D domain (note that the latter technically has a boundary, while our theoretical results required the absence of boundaries for simplicity. They still seem to be empirically valid). We use $p=2$, and compute the entropic Wasserstein barycenters with $\epsilon >0$ using a variant of Sinkhorn's algorithm \cite[Chap. 9]{peyrebook}. In Fig.~\ref{fig:curve}, we compute $\|a^\Theta-a^{\Tilde{\Theta}}\|_2^2$ on the sphere, w.r.t. $N$, and compare with the theoretical rates given by Cor.~\ref{cor:stabilityRG}. The bounds appears to be, as expected, quite loose, and the problem quite noisy.

\begin{figure}
    \centering
    \includegraphics[width=.32\textwidth]{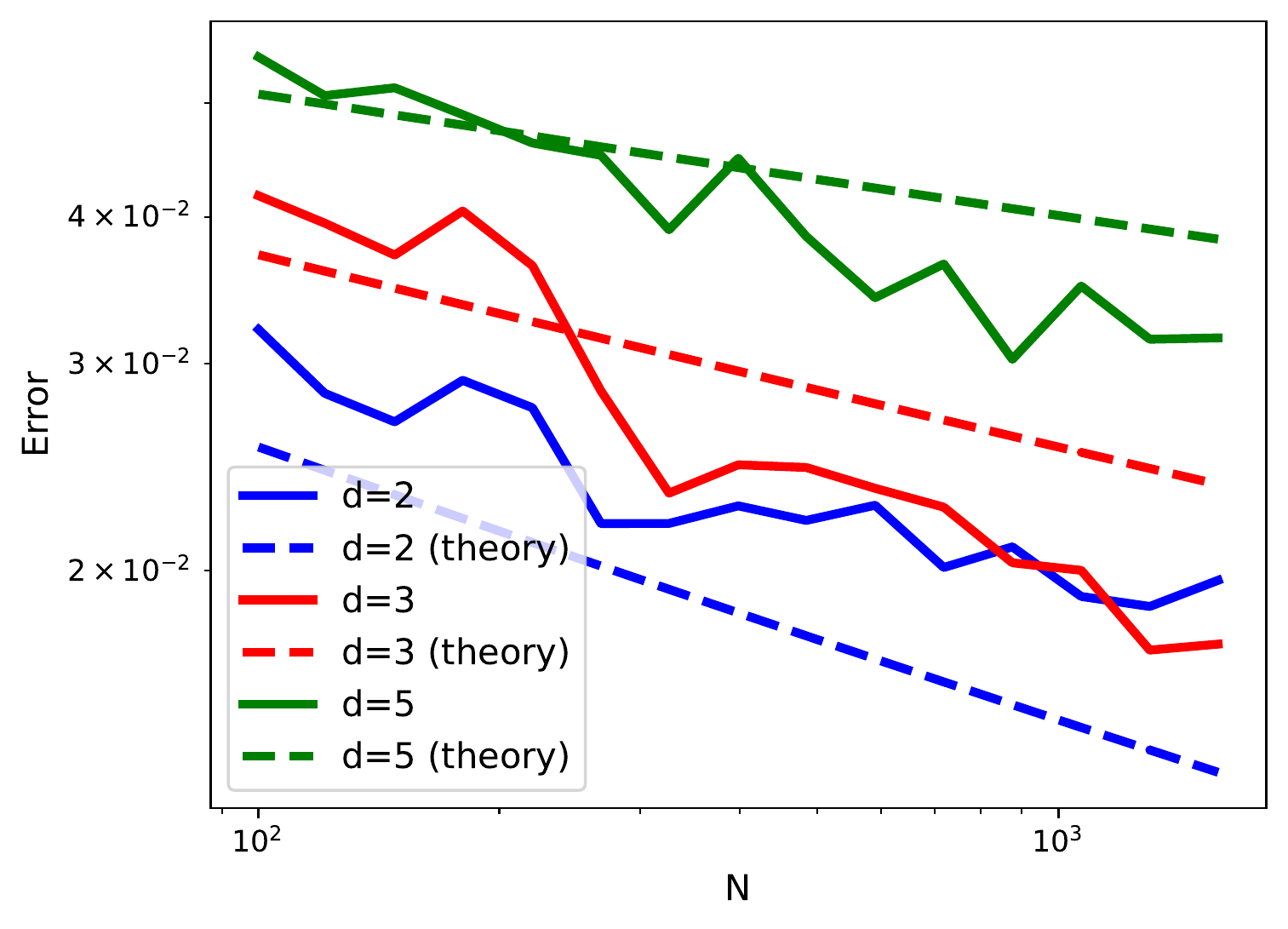}
    \caption{Error between true WBs on the (hyper-)sphere and barycenters computed with shortest paths, w.r.t. $N$, for several dimension $d$, averaged over 30 experiments.}
    \label{fig:curve}
\end{figure}







\section{Conclusion}

In this paper, we have shown the stability of entropic WBs with respect to the cost matrices. We then gave an application to random geometric graphs for which the shortest paths converge to the geodesics of the underlying manifold, guaranteeing for instance the convergence of WBs on discretized 3D shapes. 
Our theoretical work hints at many potential outlooks. 
Other models of random graphs could be treated \cite{kerivenOT}, with different applications. Finally, we have assumed fixed the supporting points of the distributions and barycenters, while the stability of WBs to sampling the target measures has recently been shown \cite{Carlier2022}. Combining the results would finalize the link between continuous \cite{Kim2017b} and discretized WBs.



\small
\bibliographystyle{IEEEbib}
\bibliography{newbiblio}

\normalsize
\appendix
\section{Proofs}

\subsection{Proof of Prop.~\ref{prop:stability}}

\begin{proof}
For all $a,b, C, \Tilde{C}$, we have
\begin{align*}
    |W_C^\epsilon(a,b) &- W_{\Tilde{C}}^\epsilon(a,b)| \\
    &\leq \sup_T |W_C^\epsilon(a,b,T) - W_{\Tilde{C}}^\epsilon(a,b,T)| \\
    &= \sup_T |<C-\Tilde{C}, T>| \\
    &\leq \|C-\Tilde{C}\|_\infty
\end{align*}
since $\sum_{ij} T_{ij}=1$. Similarly,
\begin{align*}
    |B^\epsilon(\Theta) &- B^\epsilon(\Tilde{\Theta})| \leq \sup_a|B^\epsilon(\Theta, a) - B^\epsilon(\Tilde{\Theta}, a)| \\
    &\leq \sup_a \sum_s \lambda_s |W_{C_s}^\epsilon(a,b_s) - W_{\Tilde{C}_s}^\epsilon(a,b_s)| \\
    &\leq \sum_s \lambda_s \|C-\Tilde{C}_s\|_\infty
\end{align*}
from what precedes.
\end{proof}

\subsection{Proof of Theorem \ref{thm:stability}}

Since we are solving a convex optimization problem with linear constraints, we introduce the Lagrangian:
\begin{equation}\label{eq:lagrangian}
    L_\Theta(f,g,\nu)= \mathcal{L}_\Theta(f,g) - \nu^\top \sum\limits_s \lambda_s f_s
\end{equation}
with Lagrange coefficients $\nu \in \mathbb{R}^n$. Recall that strong duality holds.
By first order conditions on the Lagrangian, the optimal $g^\Theta_s$, $f^\Theta_s, \nu^\Theta$ satisfy:

\begin{equation}\label{eq:firstorder}
\begin{cases}
    \lambda_s\left(b_s-e^{g^\Theta_s/\epsilon}\odot\left(K_s^T e^{f^\Theta_s /\epsilon}\right)\right) = 0 &\forall s\\
    \lambda_s\left(\nu^\Theta-e^{f^\Theta_s/\epsilon}\odot\left(K_s e^{g^\Theta_s /\epsilon}\right)\right) = 0 &\forall s\\
    \sum\limits_s \lambda_s f^\Theta_s = 0
\end{cases}
\end{equation}
which can be rewritten as:
\begin{equation}\label{eq:systeme}
    \begin{cases}
        b_{sj} = \sum_i K_{sij} e^{(f^\Theta_{si} + g^\Theta_{sj}) /\epsilon} &\forall s,j\\
        \nu_i^\Theta = \sum_j K_{sij} e^{(f^\Theta_{si} + g^\Theta_{sj}) /\epsilon} &\forall i\\
        \sum_s \lambda_s f^\Theta_s = 0\\
    \end{cases}
\end{equation}

We can then prove Lemma \ref{lem:borne}.
\begin{proof}[Proof of Lemma \ref{lem:borne}]
For all $s$, from \eqref{eq:systeme}, we get:
\begin{align*}
    \sum_j e^{g^\Theta_{sj} /\epsilon} &= \sum_j \frac{b_{sj}}{\sum\limits_i K_{sij}e^{f^\Theta_{si}/\epsilon}}\\
    &\leq \sum_j \frac{b_{sj}}{\sum_i \delta_{\min} e^{f_{si}^\Theta/\epsilon}}\\
    &\leq \frac{1}{\delta_{\min} \sum_i e^{f_{si}^\Theta/\epsilon}}
\end{align*}

In the same manner, $\sum_j e^{g^\Theta_{sj} /\epsilon} \geq \frac{1}{\delta_{\max} \sum_i e^{f_{si}^\Theta/\epsilon}}$ and therefore:
$$
\frac{1}{\delta_{\max} }\leq \sum_i e^{f^\Theta_{si} /\epsilon} \sum_j e^{g^\Theta_{sj} /\epsilon}  \leq  \frac{1}{\delta_{\min}}\, .
$$
\end{proof}

We then prove Lemma \ref{lem:stability}.

\begin{proof}[Proof of Lemma \ref{lem:stability}]

The function $\phi:x\rightarrow e^{x/\epsilon}$ is $e^{a/\epsilon}/\epsilon^2$-strongly convex on $[a,b]$, and therefore for $0\leq t \leq 1$:
\begin{align*}
    \phi(tx+(1-t)x') &\leq t\phi(x)+(1-t)\phi(x')\\
    &\qquad-t(1-t)e^{a/\epsilon}|x-x'|^2 / \epsilon^2
\end{align*}

Let us denote $\pi^\Theta = (f_s^\Theta, g_s^\Theta)_s$ and $\Pi^\Theta = (\pi^\Theta, \nu^\Theta)$. We have:
\begin{align*}
    &\mathcal{L}_\Theta\left( t\pi^{\tilde{\Theta}}+(1-t) \pi^\Theta \right) - t\mathcal{L}_\Theta(\pi^{\tilde{\Theta}})-(1-t)\mathcal{L}_\Theta(\pi^\Theta)\\
    &= -\epsilon \sum_{s,i,j} \lambda_s K_{sij} \Big(  e^{(t(f^{\tilde{\Theta}}_{si}+g_{sj}^{\tilde{\Theta}})+(1-t)(f_{si}^\Theta+g_{sj}^\Theta))/\epsilon} \\
    &\qquad \qquad \qquad - te^{(f_{si}^{\tilde{\Theta}} +g_{sj}^{\tilde{\Theta}}) /\epsilon} -(1-t)e^{(f_{si}^\Theta+g_{sj}^\Theta)/\epsilon} \Big)\\
    &\geq \epsilon\sum_{s,i,j} \frac{\lambda_s t(1-t) K_{sij}}{\delta_{\max} 2\epsilon^2}\left|f_{si}^{\tilde{\Theta}}+g_{sj}^{\tilde{\Theta}}-(f_{si}^\Theta+g_{sj}^\Theta)\right|^2
\end{align*}
using the strong convexity of $\phi$, since from Lemma \ref{lem:borne} we have $f^\Theta_{si}+g^\Theta_{sj} \geq \epsilon \log 1/\delta_{\max}$ and similarly for $\Tilde{\Theta}$.

Moreover, since $f^\Theta$ and $f^{\tilde{\Theta}}$ satisfy $\sum_s \lambda_s f_s=0$, it holds that $\mathcal{L}_\Theta(t\pi^{\tilde{\Theta}}+(1-t) \pi^\Theta) = L_\Theta(t\Pi^{\tilde{\Theta}}+(1-t) \Pi^\Theta)$ for all $t$. Hence, by dividing the above inequality by $t$ and taking the limit $t \to 0$,
\begin{align*}
    &\frac{ \mathcal{L}_\Theta\left(t\pi^{\tilde{\Theta}}+(1-t)\pi^\Theta\right) - \mathcal{L}_\Theta(\pi^\Theta)}{t} \\
    &\qquad = \frac{ L_\Theta\left(t\Pi^{\tilde{\Theta}}+(1-t)\Pi^\Theta\right) - L_\Theta(\Pi^\Theta)}{t} \\
    &\qquad \xrightarrow[t\to 0]{} \nabla L_\Theta(\Pi^\Theta)^\top (\Pi^{\tilde \Theta} - \Pi^\Theta) = 0
\end{align*}
since $\nabla L_\Theta(\Pi^\Theta) = 0$ by first-order conditions.
\end{proof}

Now we can prove the theorem.

\begin{proof}[Proof of Theorem \ref{thm:stability}]
Recall that $L_\Theta(\Pi^\Theta) = \mathcal{L}_\Theta(\pi^\Theta)$ and similarly for $\Tilde{\Theta}$. Following Lemma \ref{lem:stability}, we seek to bound $\mathcal{L}_\Theta(\pi^\Theta) - \mathcal{L}_\Theta(\pi^{\Tilde{\Theta}})$:
\begin{align*}
    &L_\Theta(\Pi^\Theta) - L_\Theta(\Pi^{\tilde{\Theta}})\\
    &\quad =  L_\Theta(\Pi^\Theta) - L_{\tilde{\Theta}}(\Pi^\Theta) + L_{\tilde{\Theta}}(\Pi^\Theta) - L_{\tilde{\Theta}}(\Pi^{\tilde{\Theta}})\\
    &\qquad+ L_{\tilde{\Theta}}(\Pi^{\tilde{\Theta}}) - L_\Theta(\Pi^{\tilde{\Theta}})\\
    &\quad \leq 2\sup_{\Pi}|L_\Theta(\Pi)-L_{\tilde{\Theta}}(\Pi)|
\end{align*}
since $L_{\tilde{\Theta}}(\Pi^\Theta) - L_{\tilde{\Theta}}(\Pi^{\tilde{\Theta}}) \leq 0$ by optimality, where the supremum in the last line is over the $\Pi$ that satisfy the bounds in Lemma \ref{lem:borne}.
Hence, by Lemma \ref{lem:stability}:
\begin{align*}
    &\sum\limits_{s,i,j}\lambda_s K_{sij}\left|f^{\tilde{\Theta}}_{si}+g^{\tilde{\Theta}}_{sj}-(f_{si}^\Theta+g_{sj}^\Theta)\right|^2\\
    &\leq 4\epsilon^2\delta_{\max} \sup_{f,g}\sum\limits_{s,i,j} \lambda_s e^{(f_{si}+g_{sj})/\epsilon}|\tilde{K}_{sij}-K_{sij}|\\
    &\leq 4\epsilon^2 \frac{\delta_{\max}}{\delta_{\min}}\sum_s \lambda_s\Vert K_s-\tilde{K}_s\Vert_\infty
\end{align*}
Since $\sum_{ij} e^{(f_{si}+g_{sj})/\epsilon} \leq 1/\delta_{\min}$ according to Lemma \ref{lem:borne}.
Let us recall that for all $s$, we have 
$$
a^\Theta = diag(e^{f^\Theta_s/\epsilon})K_s diag(e^{g^\Theta_s/\epsilon}) 1_{m_s}
$$
and $\sum_s \lambda_s = 1$. Therefore, using $(t+u)^2 \leq 2(t^2+u^2)$:
\begin{align*}
    &\| a^\Theta- a^{\Tilde{\Theta}}\|_2^2 = \sum_s \lambda_s \|a^\Theta-a^{\Tilde{\Theta}}\|_2^2 \\
    &\quad = \sum_{s,i} \lambda_s \vert\sum_j e^{(f^{\Theta}_{si}+g^{\Theta}_{sj})/\epsilon}K_{sij} - e^{(f^{\tilde \Theta}_{si}+g^{\tilde \Theta}_{sj})/\epsilon}\tilde K_{sij}\vert^2 \\
    &\quad \leq  2\Big(\sum_{s,i,j} \lambda_s  e^{2(f^{\tilde \Theta}_{si}+g^{\tilde \Theta}_{sj})/\epsilon}\vert K_{sij} - \tilde K_{sij}\vert^2 \\
    &\qquad + \sum_{s,i,j} \lambda_s \vert (e^{(f^{\Theta}_{si}+g^{\Theta}_{sj})/\epsilon} -e^{(f^{\tilde \Theta}_{si}+g^{\tilde \Theta}_{sj})/\epsilon}) K_{sij}\vert^2\Big)
\end{align*}
by triangular inequality.
For the first term, we have directly
\begin{align*}
     \sum_{s,i,j} &\lambda_s e^{2(f^{\Theta}_{si}+g^{\Theta}_{sj})/\epsilon}\vert K_{sij} - \tilde K_{sij} \vert^2\\
     &\leq \sum_{s} \lambda_s \Big( \sum_{ij}e^{(f^{\Theta}_{si}+g^{\Theta}_{sj})/\epsilon}\vert K_{sij} - \tilde K_{sij} \vert\Big)^2 \\
     &\leq \frac{1}{\delta_{\min}^2}\sum_s \lambda_s \|K_s- \tilde K_s\|_\infty^2
\end{align*}
Using Lemma \ref{lem:borne}.
For the second term, using intermediate value theorem combined to lemma \ref{lem:borne} and Cauchy-Schwartz inequality, we prove :
\begin{align*}
    &\sum_{s,i,j} \lambda_s \vert (e^{(f^{\Theta}_{si}+g^{\Theta}_{sj})/\epsilon} -e^{(f^{\tilde \Theta}_{si}+g^{\tilde \Theta}_{sj})/\epsilon}) K_{sij}\vert^2 \\
    &\quad \leq \frac{\delta_{\max}}{\delta_{\min}^2 \epsilon^2}\sum_{s,i,j} \lambda_s \vert f^{\Theta}_{si}+g^{\Theta}_{sj} - f^{\tilde \Theta}_{si}+g^{\tilde \Theta}_{sj}\vert^2 K_{sij} \\
    &\quad \leq \frac{1}{\delta_{\min} \epsilon} \sqrt{4\epsilon^2 \frac{\delta_{\max}}{\delta_{\min}}\sum_s \lambda_s\Vert K_s-\tilde{K}_s\Vert_\infty} \sqrt{\delta_{\max}} \\
    &\quad \leq \frac{4\delta_{\max}^2}{\delta_{\min}^3} \sum_s \lambda_s \|K_s- \tilde K_s\|_\infty
\end{align*}
which is strictly greater than twice the first term since $\|K_s- \tilde K_s\|_\infty \leq 2$. We conclude with the mean value theorem to obtain $\|C_s-\tilde C_s\|_\infty$.
\end{proof}

\subsection{Proof of Corollary \ref{cor:stabilityRG}}
\begin{proof}
By an application of Lemma \ref{lem:geodesic} and a union bound over all the pairs $(x_i, y_{sj})$, with probability at least $1-\rho$: for all $s,i,j$,
\begin{equation*}
    |d(x_i,y_{sj}) - h_N SP(x_i,y_{sj})| \lesssim h_N + \left( \frac{\log \frac{n\sum_s m_s}{h_N \rho}}{c_z N h_N^k}\right)^{1/k}
\end{equation*}
In particular, for $N$ sufficiently big, 
$$
\sup_{s,i,j} h_N SP(x_i,y_{sj}) \leq 2 D_{\mathcal{M}}
$$
Then, we use the mean value theorem to obtain that $|c^p - d^p| \leq |c-d| p \sup_{e \in [c,d]} |e|^{p-1}$, such that
\begin{equation*}
    \|C_s - \tilde C_s\|_\infty \lesssim \sup_{i,j} |d(x_i,y_{sj}) - h_N SP(x_i,y_{sj})| \cdot p D_{\mathcal{M}}^{p-1}
\end{equation*}
Finally, we conclude by applying Theorem \ref{thm:stability} with $c_{\min}=0$ and $c_{\max} = 2 D_{\mathcal{M}}$.
\end{proof}
\end{document}